\documentclass[useAMS,usenatbib]{mn2e}

\usepackage{epsfig}
\usepackage{multirow}
\usepackage{color}

\def\gsim{\;\rlap{\lower 2.5pt
 \hbox{$\sim$}}\raise 1.5pt\hbox{$>$}\;}
\def\lsim{\;\rlap{\lower 2.5pt
   \hbox{$\sim$}}\raise 1.5pt\hbox{$<$}\;}
\def\ie{{\it i.e.}}
\def\eg{{\it e.g.}}

\def\vtl{\vec{x}^L}
\def\vts{\vec{x}^S}

\def\ima{{\mathrm i}}

\title[Cosmic Shears Should Not Be Measured In Conventional Ways]{Cosmic Shears Should Not Be Measured In Conventional Ways}

\author[Jun Zhang, Eiichiro Komatsu]
{Jun Zhang$^{1, 2}$\thanks{E-mail:jzhang@astro.as.utexas.edu}, Eiichiro Komatsu$^{1}$\\ 
\\
$^{1}$Texas Cosmology Center, University of Texas at Austin, Austin, TX 78712, USA\\
$^{2}$Department of Astronomy, University of California, Berkeley, CA 94720, USA\\
}

\newtheorem{theorem}{Theorem}[section]
\newtheorem{lemma}[theorem]{Lemma}

\begin{document}


\pagerange{\pageref{firstpage}--\pageref{lastpage}} \pubyear{2010}

\maketitle

\label{firstpage}
                                      
\begin{abstract}

A long standing problem in weak lensing is about how to construct cosmic
 shear estimators from galaxy images. Conventional methods average over a single quantity per galaxy to estimate each shear component. We show that any such shear estimators must reduce to a highly nonlinear form when the galaxy image is described by three parameters (pure ellipse), even in the absence of the point spread function (PSF). In the presence of the PSF, we argue that this class of shear estimators
 do not likely exist. Alternatively, we propose a new way of measuring
 the cosmic shear: instead of averaging over a single value from each galaxy, we average over two numbers, and then take the ratio to estimate the shear component. In particular, the two numbers correspond to the numerator and denominators which generate the quadrupole moments of the galaxy image in Fourier space, as proposed in Zhang (2008). This yields a statistically unbiased estimate of the shear component. Consequently, measurements of the n-point spatial
 correlations of the shear fields should also be modified: one needs to
 take the ratio of two correlation functions to get the desired, 
 unbiased shear correlation. 

\end{abstract}


\section{Introduction}
\label{intro}
  
Weak gravitational lensing refers to the weak and systematic shape
distortions of background source images (galaxies, CMB, etc.)  by the
foreground inhomogeneous density distributions on cosmological
scales. Since this effect only involves gravity, it has been widely used
as a direct probe of matter density fluctuations of our Universe (see,
\eg, \citealt{hj08} for a recent review). 

The weak lensing effect can only be probed statistically due to the fact
that the intrinsic projected shape of each galaxy is always somewhat
random and anisotropic. A central theme in the study of 
weak lensing is to find unbiased cosmic shear estimators on galaxy
images. This is indeed very challenging because the shape distortion due
to weak lensing is generally much weaker than the intrinsic variations
of the galaxy shapes. 

In the early stage of this field, most of the work focused 
on issues regarding the use of the quadrupole moments of a
galaxy image as a shear estimator
(\citealt{tyson90,bonnet95,kaiser95,luppino97,hoekstra98,rhodes00,kaiser00}). Ever
since then, a number of other shear estimators have been considered in
the literature, including moments defined by a certain set of orthogonal
functions
(\citealt{bridle01,bernstein02,refregierbacon03,massey05,nakajima07}),
the spatial derivatives of the galaxy surface brightness field
(\citealt{zhang08,zhang10a}), etc.. 

Conventionally, for each shear component, the shear estimator is simply
one number derived from a galaxy image, whose statistical mean is
supposed to be equal to the true shear value, provided
that the intrinsic galaxy image is statistically
isotropic. Unfortunately, even in the absence of the PSF, we show that
such shear estimators at least do not exist in simple forms, making them
hard to use in practice for precise shear measurements (\eg, in the presence of noise). In the
presence of the PSF, we argue that such shear estimators do not likely
exist. We give reasons for the above statements in
\S\ref{not_exist}. (The readers who are just interested in
our new way of measuring shears may skip this section.)

In \S\ref{alternatives}, we present a new form of shear
estimators: instead of having only one number from each galaxy image for
each shear component, one can keep {\it two} numbers, and use the ratio
of their averages over many galaxies to accurately measure the cosmic
shear. We find that this new way of measuring the shear can be easily 
implemented by using the method of Zhang (2008) (Z08
hereafter). The new type of shear estimator requires weak lensing statistics such as the n-point correlation functions of the shear field to be carried out in a slightly unusual way, but with little additional cost. This is discussed in \S\ref{statistics}. In \S\ref{examples}, we give numerical examples. Finally, we summarize in \S\ref{summary}.

\section{Conventional Shear Estimators}
\label{not_exist}

Conventional shear estimators are defined as a class of shear estimators, which average over a single quantity per galaxy to estimate each shear component. Most of the existing shear measurement methods belong to this class. For example, the method by Kaiser et al. (1995) and its extensions basically use the quadrupole moments of each Gaussian-Profile-Weighted galaxy image to measure the shear components; in the shapelets method (\citealt{refregier03}), the value of each shear component is estimated from best-fitting a shapelets model to each observed galaxy image; the method of Bernstein and Jarvis (2002) evaluate the shear components by fitting each galaxy shape (also the PSF) with a series of orthogonal 2D Gaussian-based functions (see, \eg, \citealt{massey07} for more examples). A common feature of these methods is that they all generate one quantity per galaxy for each shear component.

In \S\ref{classic}, we start our discussion with the well-known 
examples of shear estimators consisting of quadrupole
moments of galaxy images, and show what the issues are. 
In \S\ref{no_PSF}, we show that, even in the absence of the PSF, any
conventional shear estimator has to reduce to a {\it highly} nonlinear
form, making it hard to use in practice. In \S\ref{prove}, we provide arguments as to why we think that
conventional shear estimators do not likely exist when a PSF is
present.

\subsection{A Review of the Problem}
\label{classic}

The use of galaxy quadrupole moments as shear estimators has been a
central topic in weak lensing for many years. It is therefore easier to
start our discussion with the quadrupole moments. To present
the issues clearly, let us first consider the case without the PSF or any photon noise. For convenience, we use $(x_1, x_2)$ or $(x, y)$ instead of $(\theta_x, \theta_y)$ for coordinates in 2D in this paper.

Suppose that the surface brightness field of the lensed galaxy image is  $f_L(\vtl)$ on the image plane, and that of the original (pre-lensing)  galaxy image is $f_S(\vts)$ on the source plane, where $\vtl$ and $\vts$ are the position angles on the image and source planes, respectively. We have the following relations:
\begin{eqnarray}
\label{fifstits} 
&&f_L(\vtl)=f_S(\vts)\\ \nonumber
&&\vtl=\mathbf{A}\vts
\end{eqnarray}
where $\mathbf{A}_{ij}=\delta_{ij}+\Phi_{ij}$, and $\Phi_{ij}=\partial x^L_i/\partial x^S_j-\delta_{ij}$, which are the spatial derivatives of the lensing deflection angle. $\Phi_{ij}$ can also be written as $\partial_{x_i}\partial_{x_j}\Phi$, where $\Phi$ is sometimes called the lensing potential. Matrix $\mathbf{A}$ can be alternatively written in terms of the convergence $\kappa=(\Phi_{11}+\Phi_{22})/2$ and the two shear components $\gamma_1=(\Phi_{11}-\Phi_{22})/2$ and $\gamma_2=\Phi_{12}$. 

The quadrupole moments of the lensed galaxy image are defined as follows:
\begin{equation}
\label{q_moments}
Q_{ij}=\int d^2\vec{x}x_ix_jf_L(\vec{x})
\end{equation}
where the origin of the coordinates has been chosen to be the center of the light, \ie, 
\begin{equation}
\label{origin}
\int d^2\vec{x}\vec{x}f_L(\vec{x})=0
\end{equation}
Let us also define the ellipticities of the image as:
\begin{eqnarray}
\label{epsilon}
\epsilon_1&=&\frac{Q_{11}-Q_{22}}{Q_{11}+Q_{22}}\\ \nonumber
\epsilon_2&=&\frac{2Q_{12}}{Q_{11}+Q_{22}}\\ \nonumber
\end{eqnarray}

In the absence of the PSF, the quantities $\epsilon_1$ and $\epsilon_2$
are often thought to be good estimators for $\gamma_1$ and $\gamma_2$ up
to the first order in the shear. Let us find out if they are indeed unbiased shear estimators. 
The observed quadrupole can be rewritten from using eq.(\ref{fifstits}) in eq.(\ref{q_moments}):
\begin{eqnarray}
\label{q2}
Q_{ij}&=&\int d^2\vec{x}x_ix_jf_S(\mathbf{A}^{-1}\vec{x})\\ \nonumber
&=&\vert \mathrm{det} (\mathbf{A})\vert\int d^2\vec{x}(\mathbf{A}\vec{x})_i(\mathbf{A}\vec{x})_jf_S(\vec{x})
\end{eqnarray}
Note that the last step of the above equation is achieved by redefining $\mathbf{A}^{-1}\vec{x}$ as $\vec{x}$. 
Keeping up to first order in $\kappa$, $\gamma_1$, and $\gamma_2$ in eq.(\ref{q2}), we get:
\begin{eqnarray}
\label{q3}
Q_{11}-Q_{22}&=&(1+4\kappa)(Q_{11}^S-Q_{22}^S)+2\gamma_1(Q_{11}^S+Q_{22}^S)\\ \nonumber
Q_{12}&=&(1+4\kappa)Q_{12}^S+\gamma_2(Q_{11}^S+Q_{22}^S)\\ \nonumber
Q_{11}+Q_{22}&=&(1+4\kappa)(Q_{11}^S+Q_{22}^S)+2\gamma_1(Q_{11}^S-Q_{22}^S)\\ \nonumber
&+&4\gamma_2Q_{12}^S
\end{eqnarray}
where $Q_{ij}^S$ are the quadrupole moments of the original galaxy image defined as:
\begin{equation}
\label{q_I_moments}
Q_{ij}^S=\int d^2\vec{x}x_ix_jf_S(\vec{x})
\end{equation}
Note that the two light centers defined in the image and source planes coincide. Based on eq.(\ref{q3}), we find:
\begin{eqnarray}
\label{q4}
\epsilon_1&=&\frac{\epsilon_1^S+2\gamma_1}{1+2\gamma_1\epsilon_1^S+2\gamma_2\epsilon_2^S}\\ \nonumber
&=&\epsilon_1^S+2\gamma_1\left[1-\left(\epsilon_1^S\right)^2\right]-2\gamma_2\epsilon_1^S\epsilon_2^S\\ \nonumber
\epsilon_2&=&\frac{\epsilon_2^S+2\gamma_2}{1+2\gamma_1\epsilon_1^S+2\gamma_2\epsilon_2^S}\\ \nonumber
&=&\epsilon_2^S+2\gamma_2\left[1-\left(\epsilon_2^S\right)^2\right]-2\gamma_1\epsilon_1^S\epsilon_2^S
\end{eqnarray}
where 
\begin{eqnarray}
\label{epsilon_S}
\epsilon_1^S&=&\frac{Q_{11}^S-Q_{22}^S}{Q_{11}^S+Q_{22}^S}\\ \nonumber
\epsilon_2^S&=&\frac{2Q_{12}^S}{Q_{11}^S+Q_{22}^S}
\end{eqnarray}
Given that the surface brightness distribution of the original galaxy
image is statistically isotropic, we have
$\langle\epsilon_{1,2}^S\rangle=0$ and
$\langle\epsilon_1^S\epsilon_2^S\rangle=0$. Therefore, we find
\begin{eqnarray}
\label{q5}
\langle\epsilon_1\rangle&=&2\gamma_1\left[1-\langle\left(\epsilon_1^S\right)^2\rangle\right]\\ \nonumber
\langle\epsilon_2\rangle&=&2\gamma_2\left[1-\langle\left(\epsilon_2^S\right)^2\rangle\right]
\end{eqnarray}
This result, Eq.(\ref{q5}), clearly shows that $\epsilon_1$ and $\epsilon_2$ are
{\it not} unbiased shear estimators, as
$\langle\left(\epsilon_1^S\right)^2\rangle$ and
$\langle\left(\epsilon_2^S\right)^2\rangle$ in the multiplicative
factors depend on the galaxy morphology distribution, and cannot be
reduced to constant factors. (Also see Eq.(3.29) of Bernstein \& Jarvis (2002), or Eq.(9.5.26) of
Weinberg (2008).)

One can construct an unbiased estimator of the shear, 
if one keeps three quantities from each lensed galaxy image:
$Q_{11}-Q_{22}$, $2Q_{12}$, and $Q_{11}+Q_{22}$, and use the ratios of
their averages. Assuming statistical isotropy of intrinsic galaxy shapes in eq.(\ref{q3}), and keeping up to first order in shear/convergence, we have 
(also see Eq.(9.5.30) of Weinberg (2008)):
\begin{eqnarray}
\label{q6}
\frac{1}{2}\frac{\langle Q_{11}-Q_{22}\rangle}{\langle Q_{11}+Q_{22}\rangle}&=&\gamma_1\\ \nonumber
\frac{\langle Q_{12}\rangle}{\langle Q_{11}+Q_{22}\rangle}&=&\gamma_2
\end{eqnarray}
This form of shear estimators is not conventional, as one has to keep
more than one quantities from each galaxy image for each shear
component. It is this class of estimators we shall discuss
in this paper in detail.

One may wonder whether unbiased shear estimators in the conventional
form ever exist. The answer is yes, at least when the PSF is 
absent. For example, we find the following unbiased shear estimators:
\begin{eqnarray}
\label{q7}
\frac{1}{4}\left\langle \ln\left(\frac{1+\epsilon_1}{1-\epsilon_1}\right)\right\rangle&=&\gamma_1\\ \nonumber
\frac{1}{4}\left\langle \ln\left(\frac{1+\epsilon_2}{1-\epsilon_2}\right)\right\rangle&=&\gamma_2\\ \nonumber
\end{eqnarray}
Eq.(\ref{q7}) can be checked by applying Taylor expansion of $\ln[(1+\epsilon_i)/(1-\epsilon_i)]$ ($i=1, 2$) to the first order in shear/convergence using eq.(\ref{q4}).
Eq.(\ref{q7}) defines a special type of conventional shear estimators
that are accidentally found by us. It is now immediately interesting to
ask if there exist other types of unbiased shear estimators in the
conventional form. We study this issue specifically in the next two
sections. If the readers wish to go directly to the
relevant sections on the new estimator, read on from \S\ref{alternatives}.

For notational convenience, we shall abbreviate ``conventional shear
estimator'' as ``CSE'' in the rest of the paper. Once again,
by CSE we mean the shear estimators that are made of just one number measured
from a galaxy image for each shear component.

\subsection{CSE in the Absence of the PSF}
\label{no_PSF}

In preparation for our main theme of this section, we discuss the spin properties of cosmic shears and their estimators in \S\ref{coor_spin}. We then study the forms of CSEs in the absence of the PSF in \S\ref{spin_2_esti}.

\subsubsection{The Spin of the CSE}
\label{coor_spin}

To study the forms of the CSEs, it is useful to first consider their properties under coordinate rotations. Suppose we rotate the coordinates $(x, y)$ clockwise by an angle $\theta$. The new coordinates $(x^{\theta}, y^{\theta})$ are related to the old one via the following relation:
\begin{eqnarray}
\label{coor_trans}
&&x^{\theta}=x\cos\theta-y\sin\theta \\ \nonumber
&&y^{\theta}=x\sin\theta+y\cos\theta
\end{eqnarray}
If we write the position vector as a complex number of the form $x+\ima y$, where $\ima$ is the complex unit, the coordinate transformation under rotation can then be written as:
\begin{equation}
\label{coor_trans_complex}
x^{\theta}+\ima y^{\theta} =(x+\ima y)\exp(\ima \theta)
\end{equation} 
For notational brevity, we shall generally use $X^{\theta}$ to denote the value of any quantity $X$ in the new coordinates that are rotated clockwise by an angle $\theta$ with respect to the original coordinates. 

Let us now discuss how cosmic shears and their CSEs transform under
coordinate rotation. The definitions of the shear components (shown in
the beginning of \S\ref{classic}) involve spatial derivatives; 
thus, their transformation rules under coordinate rotation
can be found from the chain rule:
\begin{eqnarray}
\label{chain_rule}
&&\frac{\partial}{\partial x}=\frac{\partial x^{\theta}}{\partial x}\frac{\partial}{\partial x^{\theta}}+\frac{\partial y^{\theta}}{\partial x}\frac{\partial}{\partial y^{\theta}}\\ \nonumber
&&\frac{\partial}{\partial y}=\frac{\partial x^{\theta}}{\partial y}\frac{\partial}{\partial x^{\theta}}+\frac{\partial y^{\theta}}{\partial y}\frac{\partial}{\partial y^{\theta}}
\end{eqnarray}
From eq.(\ref{coor_trans}), we get:
\begin{eqnarray}
\label{cooefi_trans_derivs}
&&\frac{\partial x^{\theta}}{\partial x}=\frac{\partial y^{\theta}}{\partial y}=\cos\theta\\ \nonumber
&&\frac{\partial y^{\theta}}{\partial x}=-\frac{\partial x^{\theta}}{\partial y}=\sin\theta 
\end{eqnarray}
Therefore, we have:
\begin{equation}
\label{coor_trans_complex_2}
\frac{\partial}{\partial x^{\theta}}+\ima\frac{\partial}{\partial y^{\theta}}=\left(\frac{\partial}{\partial x}+\ima\frac{\partial}{\partial y}\right)\exp(\ima\theta)
\end{equation}
Taking the square of eq.(\ref{coor_trans_complex_2}), we find:
\begin{eqnarray}
\label{coor_trans_complex_3}
&&\left(\frac{\partial^2}{\partial (x^{\theta})^2}-\frac{\partial^2}{\partial (y^{\theta})^2}\right)+\ima\left(2\frac{\partial^2}{\partial x^{\theta}\partial y^{\theta}}\right)\\ \nonumber
&=&\left[\left(\frac{\partial^2}{\partial x^2}-\frac{\partial^2}{\partial y^2}\right)+\ima\left(2\frac{\partial^2}{\partial x\partial y}\right)\right]\exp(\ima 2\theta)
\end{eqnarray}
Therefore, the shear components, which are 2nd order derivatives of the lensing potential $\Phi$, also transform under coordinate rotation as:
\begin{equation}
\label{gamma_rot}
\gamma_1^{\theta}+\ima\gamma_2^{\theta}=\left(\gamma_1+\ima\gamma_2\right)\exp(\ima 2\theta)
\end{equation}
Because of this property, we usually say that cosmic shears form a spin-2 quantity. In general, a complex quantity, say $\Pi$, is called a spin-$n$ quantity if it transforms as $\Pi^{\theta}=\Pi\exp(\ima n\theta)$ under a clockwise coordinate rotation of angle $\theta$.

Now let us discuss shear estimators. It is straightforward to see that the CSEs defined in eq.(\ref{q7}) do not form a spin-2 quantity. More generally, assuming that $\Gamma_1$ and $\Gamma_2$ are the CSEs for $\gamma_1$ and $\gamma_2$, respectively, then, unlike $\gamma_1+\ima\gamma_2$, $\Gamma_1+\ima\Gamma_2$ is not necessarily a spin-2 quantity.  However, it turns out that we can regularize any CSEs by turning them into components of a spin-2 quantity. We have the following lemma:

\begin{lemma}
\label{turn_spin_2}
Based on any pair of CSEs $(\Gamma_1, \Gamma_2)$, one can build a new pair of CSEs $(\Gamma_1', \Gamma_2')$ to form a spin-2 quantity through the following procedure:
\begin{equation}
\label{trans_spin_2}
\Gamma_1'+\ima\Gamma_2'=\frac{1}{2\pi}\int_0^{2\pi}d\theta\exp(-\ima 2\theta)\left(\Gamma_1^{\theta}+\ima\Gamma_2^{\theta}\right)
\end{equation}  
\end{lemma}

\begin{proof}
Firstly, under a clockwise coordinate rotation by angle $\theta_0$, we have:
\begin{eqnarray}
\label{trans_spin_2_rot}
&&\Gamma_1'^{\theta_0}+\ima\Gamma_2'^{\theta_0}\\ \nonumber
&=&\frac{1}{2\pi}\int_0^{2\pi}d\theta\exp(-\ima 2\theta)\left(\Gamma_1^{\theta_0+\theta}+\ima\Gamma_2^{\theta_0+\theta}\right)\\ \nonumber
&=&\frac{1}{2\pi}\int_0^{2\pi}d\theta'\exp[-\ima 2(\theta'-\theta_0)]\left(\Gamma_1^{\theta'}+\ima\Gamma_2^{\theta'}\right)\\ \nonumber
&=&\left(\Gamma_1'+\ima\Gamma_2'\right)\exp(\ima 2\theta_0)
\end{eqnarray}  
Therefore, $\Gamma_1'+\ima\Gamma_2'$ form a spin-2 quantity. To show that $(\Gamma_1', \Gamma_2')$ are a pair of CSEs, let us take the ensemble average on both sides of eq. (\ref{trans_spin_2_rot}): 
\begin{eqnarray}
\label{trans_ideal_esti}
&&\left\langle\Gamma_1'+\ima\Gamma_2'\right\rangle\\ \nonumber
&=&\frac{1}{2\pi}\int_0^{2\pi}d\theta\exp(-\ima 2\theta)\left\langle\Gamma_1^{\theta}+\ima\Gamma_2^{\theta}\right\rangle
\end{eqnarray}  
Since by definition, $\Gamma_1^{\theta}$ and $\Gamma_2^{\theta}$ measure the shear values in the rotated coordinates, we have,
\begin{eqnarray}
\label{Gamma_theta}
&&\left\langle\Gamma_1^{\theta}\right\rangle=\gamma_1^{\theta}=\gamma_1\cos 2\theta-\gamma_2\sin 2\theta\\ \nonumber
&&\left\langle\Gamma_2^{\theta}\right\rangle=\gamma_2^{\theta}=\gamma_1\sin 2\theta+\gamma_2\cos 2\theta
\end{eqnarray}
In a more compact form, we can write eq.(\ref{Gamma_theta}) as:
\begin{equation}
\label{Gamma_theta2}
\left\langle\Gamma_1^{\theta}+\ima\Gamma_2^{\theta}\right\rangle=(\gamma_1+\ima\gamma_2)\exp(\ima 2\theta)
\end{equation}
Using eq. (\ref{Gamma_theta2}) in eq. (\ref{trans_ideal_esti}), we get:
\begin{equation}
\label{trans_ideal_esti2}
\left\langle\Gamma_1'+\ima\Gamma_2'\right\rangle=\gamma_1+\ima\gamma_2
\end{equation}  
which proves that $(\Gamma_1', \Gamma_2')$ are indeed also a pair of CSEs. 
\end{proof}

Due to the invariance under a coordinate rotation of angle $2\pi$, one can always decompose any shear estimator into components of integer spins. Eq.(\ref{trans_spin_2}) essentially defines a procedure of isolating the spin-2 components of any CSEs using Fourier transformation. Since the cosmic shears form a spin-2 quantity, only the spin-2 components of any CSEs are the relevant/principle components of the estimators. This point is further supported by the fact that the ensemble averages of spin-$n$ ($n\ne 2$) components of any CSEs are zero, \ie,
\begin{eqnarray}
\label{trans_ideal_est3i}
&&\frac{1}{2\pi}\int_0^{2\pi}d\theta\exp(-\ima n\theta)\left\langle\Gamma_1^{\theta}+\ima\Gamma_2^{\theta}\right\rangle\\ \nonumber
&=&\frac{1}{2\pi}\int_0^{2\pi}d\theta\exp(-\ima n\theta)\left(\gamma_1+\ima\gamma_2\right)\exp(\ima 2\theta)\\ \nonumber
&=&0 \mbox{  ( if  } n\ne 2\mbox{ )}
\end{eqnarray}  
Therefore, we only need to focus on spin-2 CSEs from here on.

\subsubsection{Spin-2 CSEs}
\label{spin_2_esti}

In the weak lensing limit, \ie, when the cosmic shear parameters $(\gamma_1, \gamma_2, \kappa)$ are small, any spin-2 CSEs $(\Gamma_1, \Gamma_2)$ can be Taylor expanded to the first order in shear as follows:
\begin{eqnarray}
\label{Taylor_weak_lensing}
&&\Gamma_1(\gamma_1,\gamma_2,\kappa)\\ \nonumber
&=&\left(\Gamma_1\right)_0+\gamma_1\left(\frac{\partial\Gamma_1}{\partial\gamma_1}\right)_0+\gamma_2\left(\frac{\partial\Gamma_1}{\partial\gamma_2}\right)_0+\kappa\left(\frac{\partial\Gamma_1}{\partial\kappa}\right)_0\\ \nonumber
&&\Gamma_2(\gamma_1,\gamma_2,\kappa)\\ \nonumber
&=&\left(\Gamma_2\right)_0+\gamma_1\left(\frac{\partial\Gamma_2}{\partial\gamma_1}\right)_0+\gamma_2\left(\frac{\partial\Gamma_2}{\partial\gamma_2}\right)_0+\kappa\left(\frac{\partial\Gamma_2}{\partial\kappa}\right)_0
\end{eqnarray}
where $(X)_0$ ($X$ is any quantity) denotes the value of $X$ at $\gamma_1=\gamma_2=\kappa=0$. Since $(\Gamma_1)_0$, $(\Gamma_2)_0$, $(\partial_{\kappa}\Gamma_1)_0$, and $(\partial_{\kappa}\Gamma_2)_0$ are all spin-2 quantities, their ensemble average must vanish. On the other hand, the coefficients associated with $\gamma_1$ and $\gamma_2$ in eq.(\ref{Taylor_weak_lensing}) can be decomposed into spin-0 and spin-4 components as follows:
\begin{eqnarray}
\label{Taylor_weak_lensing2}
&&\left\langle\Gamma_1(\gamma_1,\gamma_2,\kappa)\right\rangle=\gamma_1\left\langle A+B_1\right\rangle+\gamma_2\left\langle C+B_2\right\rangle\\ \nonumber
&&\left\langle\Gamma_2(\gamma_1,\gamma_2,\kappa)\right\rangle=\gamma_1\left\langle B_2-C\right\rangle+\gamma_2\left\langle A-B_1\right\rangle
\end{eqnarray}
where 
\begin{eqnarray}
\label{ABC}
&& A=\frac{1}{2}\left(\partial_{\gamma_1}\Gamma_1+\partial_{\gamma_2}\Gamma_2\right)\\ \nonumber
&& C=\frac{1}{2}\left(\partial_{\gamma_2}\Gamma_1-\partial_{\gamma_1}\Gamma_2\right)\\ \nonumber
&& B_1=\frac{1}{2}\left(\partial_{\gamma_1}\Gamma_1-\partial_{\gamma_2}\Gamma_2\right)\\ \nonumber
&& B_2=\frac{1}{2}\left(\partial_{\gamma_2}\Gamma_1+\partial_{\gamma_1}\Gamma_2\right)
\end{eqnarray} 
As shown in Appendix A,  $A$ is a scalar, $C$ is a pseudo scalar, $B_1+\ima B_2$ is a spin-4 quantity. The ensemble averages of $B_1$ and $B_2$ must vanish. The ensemble average of $C$ vanishes if galaxy images have parity symmetry along any direction in the plane of the sky statistically, which is assumed to be true in this paper.  Consequently, for $(\Gamma_1, \Gamma_2)$ to be spin-2 CSEs, we only require $\langle A\rangle=1$. This actually implies that for any individual galaxy, $A=1$. The reason is that $A$ of any single galaxy does not change under coordinate rotation of random angles, and is equal to $\langle A\rangle$ because the galaxies generated by rotations of a single galaxy form a complete set of statistically isotropic samples (\ie, there are no special directions). As a result, any spin-2 CSEs $(\Gamma_1, \Gamma_2)$ must satisfy the following necessary condition:
\begin{equation}
\label{requirement}
\frac{\partial\Gamma_1}{\partial\gamma_1}+\frac{\partial\Gamma_2}{\partial\gamma_2}=2
\end{equation}
This is also a sufficient condition, because $A=1$ directly implies that $\langle A\rangle=1$.

In general, the CSEs are functions of a certain number of shape
parameters (\eg, the multipole moments of an image). The functions can
be very complicated, and are certainly not fixed by the requirement
given by eq.(\ref{requirement}). However, for galaxies whose shapes
are described by only three parameters (perfect ellipses), any CSE
should reduce to a function of just three variables. In this case, we
find that any pair of spin-2 CSEs must reduce to a unique form, which is
sufficient for us to judge whether CSEs are convenient in
practice: namely, if we find that the resulting form is highly non-linear even for such a simple case, then it is
reasonable to conclude that CSEs are not so useful for accurate shear measurements from more realistic
galaxy shapes as well. This is shown in the rest of this section. For clarity, we refer the readers to Appendix B for the mathematical details/proofs for some of the statements made hereafter in this section.

Let us consider a set of galaxies whose surface brightness profiles can be parametrized as $f_S(R)$ with $R=a(x^2+y^2)+b(x^2-y^2)+2cxy$, where $x$ and $y$ are the coordinates, $f_S(R)$ is a function of a fixed form, and $(a, b, c)$ are the three parameters determining galaxy shapes. For the images to be ellipses, we require the following three things: 1. $f_S(R)$ decays sufficiently fast when $R$ becomes large; 2. $a+b>0$; 3. $a^2-b^2>c^2$. For example, if $f_S(R)=H(R_c-R)$ ($H$ is the step function) and $(a, b, c)$ satisfy the above conditions, the galaxy surface brightness is then distributed evenly inside the ellipse defined by $a(x^2+y^2)+b(x^2-y^2)+2cxy\le R_c$. When such images are weakly lensed, the three conditions are not violated, and $(a, b, c)$ becomes $(a', b', c')$ without changing the form of $f_S$. In other words, weak lensing does not introduce additional degrees of freedom to the galaxy shapes. Note that otherwise, one has to consider using more than 3 parameters to construct shear estimators. Among the three parameters, there are indeed only two degrees of freedom useful for shear measurement: the ratios of the parameters. This is because the overall amplitudes of $(a, b, c)$ only change the galaxy size, not its shape. As shown in Appendix B, the ellipticities $(\epsilon_1, \epsilon_2)$ defined in \S\ref{classic} are directly equal to $(-b/a, -c/a)$, therefore, we can write the shear estimators as functions of only $\epsilon_1$ and $\epsilon_2$. 

We can further show that $(\Gamma_1, \Gamma_2)$ must take the following form:
\begin{equation}
\label{G_1_G_2_form}
\Gamma_1+\ima\Gamma_2=(\epsilon_1+\ima\epsilon_2)g(u)
\end{equation}
where $u=\epsilon_1^2+\epsilon_2^2$, and $g$ is a one-variable complex function, whose form is to be determined later in this section. To see why eq.(\ref{G_1_G_2_form}) is true, one can use the Taylor expansion to write $\Gamma_1$ and $\Gamma_2$ as power series of $\epsilon_1$ and $\epsilon_2$ \footnote{We do not consider shear estimators that cannot be Taylor expanded as power series of the galaxy shape parameters in this paper. Such shear estimators likely suffer numerical instabilities in practice.}.  Consequently, one can write $\Gamma_1+\ima\Gamma_2$ as power series of $\epsilon_1+\ima\epsilon_2$ and $\epsilon_1-\ima\epsilon_2$, whose spins are $2$ and $-2$ respectively. Since $\Gamma_1+\ima\Gamma_2$ is a spin-2 quantity, each term in the power series must also be a spin-2 quantity. Therefore, in each term of the power series, the power on $\epsilon_1+\ima\epsilon_2$ must be larger than that on $\epsilon_1-\ima\epsilon_2$ by exactly one, \ie, each term must take the form of  $\alpha(\epsilon_1+\ima\epsilon_2)(\epsilon_1^2+\epsilon_2^2)^n$, where $n$ is any non-negative integer, and $\alpha$ is a coefficient which can be any complex number at this point.  As a result,  the shear estimators must have the form defined in eq.(\ref{G_1_G_2_form}). To find out what $g(u)$ is, let us write it as $g_1(u)+\ima g_2(u)$. Eq. (\ref{G_1_G_2_form}) then becomes:
\begin{eqnarray}
\label{G_1_G_2_real}
&&\Gamma_1=\epsilon_1g_1(u)-\epsilon_2g_2(u)\\ \nonumber
&&\Gamma_2=\epsilon_1g_2(u)+\epsilon_2g_1(u)
\end{eqnarray}
Using the constraint in eq.(\ref{requirement}), we find:
\begin{eqnarray}
\label{constraint}
2&=&\frac{\partial\Gamma_1}{\partial\gamma_1}+\frac{\partial\Gamma_2}{\partial\gamma_2}\\ \nonumber
&=&\frac{\partial\Gamma_1}{\partial\epsilon_1}\frac{\partial\epsilon_1}{\partial\gamma_1}+\frac{\partial\Gamma_1}{\partial\epsilon_2}\frac{\partial\epsilon_2}{\partial\gamma_1}+\frac{\partial\Gamma_2}{\partial\epsilon_1}\frac{\partial\epsilon_1}{\partial\gamma_2}+\frac{\partial\Gamma_2}{\partial\epsilon_2}\frac{\partial\epsilon_2}{\partial\gamma_2}\\ \nonumber
&=&2(2-u)g_1(u)+4u(1-u)\frac{dg_1}{du}
\end{eqnarray}
It is interesting to note that eq.(\ref{constraint}) does not place any constraints on $g_2(u)$, \ie, it can be any real function. This is because $g_2(u)$ simply adds unnecessary even-parity terms into odd parity ones, and vice versa, without affecting the ensemble averages and the spin of the shear estimators. For convenience, we set $g_2(u)=0$ hereafter. 

Eq.(\ref{constraint}) is a typical first-order ordinary differential equation. It can be solved by introducing an integrating factor $k(u)$ which satisfies:
\begin{equation}
\label{ku}
k(u)(2-u)=\frac{d}{du}\left[2k(u)u(1-u)\right]
\end{equation}
Multiplying both sides of eq.(\ref{constraint}) with $k(u)$, we get:
\begin{equation}
\label{g_1_u}
\frac{d}{du}\left[2k(u)u(1-u)g_1(u)\right]=k(u)
\end{equation} 
It is now straightforward to solve both eq.(\ref{ku}) and eq.(\ref{g_1_u}). The results are:
\begin{eqnarray}
\label{results_g1}
&&k(u)\propto(1-u)^{-3/2}\\ \nonumber
&&g_1(u)=\frac{1}{u}\left(1+C\sqrt{1-u}\right) 
\end{eqnarray}
where $C$ is a real number constant. To guarantee that $\Gamma_1$ and $\Gamma_2$ do not diverge when $\epsilon_1$ and $\epsilon_2$ approach zero, we need $C=-1$. Finally, we find the unique form for the spin-2 CSEs: 
\begin{equation}
\label{G_1_G_2_form_final}
\Gamma_1+\ima\Gamma_2=(\epsilon_1+\ima\epsilon_2)\frac{1-\sqrt{1-\epsilon_1^2-\epsilon_2^2}}{\epsilon_1^2+\epsilon_2^2}
\end{equation}
Regarding the uniqueness, it is useful to note that if we transform the CSEs defined in  eq.(\ref{q7}) into spin-2 shear estimators using the procedure given in eq.(\ref{trans_spin_2}), we achieve the same shear estimators as those shown in eq.(\ref{G_1_G_2_form_final}). 

We have shown that the principle components (spin-2) of any pair of CSEs have to take specific and {\it highly} nonlinear forms for galaxies of elliptical shapes. This feature makes CSEs not convenient in practice (\eg, in the presence of noise).

\subsection{CSEs in the Presence of the PSF}
\label{prove}

Any CSEs which correct for the PSF effect also have to reduce to the forms given in eq.(\ref{G_1_G_2_form_final}) in the limit of zero PSF size when the galaxy images have pure elliptical shapes. For this reason, the conclusion in the previous section is already sufficient to argue against the usefulness of CSEs in practice.  For academic interests, we provide the following arguments for why CSEs may not even exist in the presence of the PSF:

In the presence of the point spread function, structural details of galaxy images on scales smaller than the size of the PSF are smeared out. This implies that there are only a finite number of shape parameters (\eg, multipole moments up to some order) available for constructing shear estimators. On the other hand, the derivatives of the lower order shape parameters (\eg, lower order multipole moments) with respect to the cosmic shears depend on the higher order shape parameters because of the PSF, suggesting the requirement for an infinite number of shape parameters to form the shear estimators. Combining the above two reasons, we find it unlikely to form CSEs when a PSF is present. The mathematical details of the above statements are given in Appendix C.

\section{A New Way of Estimating Shears}
\label{alternatives}

As searching for optimal shear estimators is actively
ongoing nowadays (\citealt{heymans06,massey07,bridle09,bridle10}), it is
important to realize that CSEs (``conventional shear
estimators,'' by which we mean the shear estimators 
that are made of just one number from a galaxy image for each shear
component) are hard to use in practice due to their unavoidable complex
forms even in the absence of the PSF (simpler forms, such as
the quadrupole moments, are biased estimators, as shown in \S~2.1).

Therefore, existing shear estimators of the conventional type must quantify the bias factor when estimating the shear, which can be achieved numerically (see, \eg, Erben et al.  2001, Bacon et al. 2001, or most recently, \citealt{heymans06,massey07,bridle10}) or estimated analytically (\eg, shear susceptibility in KSB [Kaiser et al. 1995] and derived methods, or responsivity factor in Bernstein \& Jarvis 2002 and similar methods), although most people have been mainly focusing on the systematic errors caused by the photon noise and the PSF. However, to achieve percent or even sub-percent level accuracy in cosmic shear measurements, it does not seem enough to completely rely on numerical tests using computer-generated galaxies of limited morphology richness, or approximate analytical methods. Unfortunately, in the presence of PSF, most of the existing shear measurement methods are too complicated or too model-dependent (\citealt{vb10,bernstein10}) to allow for an accurate analytic analysis of the systematic errors in their shear estimators.

The method of Z08 (see also \citealt{zhang10a} for the treatment of photon noise and the pixelation effect) is easily amenable to the corrections described in eq.(\ref{q5}), and can also account for the PSF correction. Not only is it simple, but also well
defined regardless of the morphologies of galaxies and the PSF. We show
here how to properly use this method (instead of using it as CSEs) to
recover the cosmic shear in an unbiased way.

\subsection{The Idea}

The basic idea of Z08 is to use the spatial derivatives of the galaxy
surface brightness field to measure the cosmic shears. It relies on the
fact that gravitational lensing does not only distort the overall shape
of the object, but also locally modifies the anisotropy of the gradient
field of the surface brightness. As it allows for using the
shape information from galaxy substructures, the method of Z08 can
potentially improve on the signal-to-noise ratio of the
shear measurements.

It is shown in Z08 that the shear measurement should be carried out in the Fourier space, in which any PSF can be transformed into the desired isotropic Gaussian form through multiplications, and the spatial derivatives of the surface brightness field can be easily measured. The cosmic shear can be estimated using the following relations:
\begin{eqnarray}
\label{shear12PSF}
&&\frac{1}{2}\frac{\langle (\partial_1f_O)^2-(\partial_2f_O)^2\rangle}{\langle (\partial_1f_O)^2+(\partial_2f_O)^2+\Delta\rangle}=-\gamma_1 \\ \nonumber
&&\frac{\langle\partial_1f_O\partial_2f_O\rangle}{\langle (\partial_1f_O)^2+(\partial_2f_O)^2+\Delta\rangle}=-\gamma_2
\end{eqnarray}
where
\begin{equation}
\label{Delta}
\Delta=\frac{\beta^2}{2}\vec{\nabla}f_O\cdot\vec{\nabla}(\nabla^2f_O)
\end{equation}
$\beta$ is the scale radius of the isotropic Gaussian PSF $W_{\beta}$, which is defined as:
\begin{equation}
W_{\beta}(\vec{\theta})=\frac{1}{2\pi\beta^2}\exp\left(-\frac{\vert\vec{\theta}\vert^2}{2\beta^2}\right)
\end{equation}    
$f_O$ is the surface brightness field. $\partial_i$ denotes $\partial /\partial x_i$. As shown in Appendix D, the method of Z08 effectively utilizes the quadrupole moments in the Fourier space to measure the cosmic shears.
 
\subsection{A New Unbiased Estimator}
\label{m_zhang08}

Now, here is an important point: in order to implement this
method, we must make it clear what we mean by the angular brackets in
eq.(\ref{shear12PSF}). First, we need to measure the derivatives of the
surface brightness and average them within a single galaxy. Let us
denote this averaging by $\langle\rangle_g$, and write:
\begin{eqnarray*}
&&\frac{1}{2}\frac{\langle (\partial_1f_O)^2-(\partial_2f_O)^2\rangle_g}{\langle (\partial_1f_O)^2+(\partial_2f_O)^2+\Delta\rangle_g} \\ 
&&\frac{\langle\partial_1f_O\partial_2f_O\rangle_g}{\langle (\partial_1f_O)^2+(\partial_2f_O)^2+\Delta\rangle_g}
\end{eqnarray*}
Of course, these are still extremely noisy as they use only one galaxy. The
question is then, ``how do we average these quantities over many
galaxies to obtain an unbiased estimator of the shears?''

If one uses these quantities as if they were the CSEs, then one would
simply average them over many galaxies. However, this will produce a
biased estimator:
\footnote{This bias was
not detected in the numerical calculations presented in Z08, as the number
of galaxies ($10^4$) used for the average was still too small.}  
\begin{eqnarray}
\label{shear12PSF_wrong}
&&\left\langle\frac{1}{2}\frac{\langle (\partial_1f_O)^2-(\partial_2f_O)^2\rangle_g}{\langle (\partial_1f_O)^2+(\partial_2f_O)^2+\Delta\rangle_g}\right\rangle_{en}=-\gamma_1(1-\delta_1) \\ \nonumber
&&\left\langle\frac{\langle\partial_1f_O\partial_2f_O\rangle_g}{\langle (\partial_1f_O)^2+(\partial_2f_O)^2+\Delta\rangle_g}\right\rangle_{en}=-\gamma_2(1-\delta_2)
\end{eqnarray}
where $\delta_1$ and $\delta_2$ are the ensemble averages of functions
of multipole moments of the galaxy images in Fourier space, and
$\langle\rangle_{en}$ denotes the ensemble average over many
galaxies. The derivation of the forms of $\delta_1$ and $\delta_2$ is
given in Appendix 
D. $\delta_1$ and $\delta_2$ are generally nonzero and dependent on the
galaxy morphology.

Instead, we need to take the ensemble averages of the
numerator and the denominator separately first, and then divide them to obtain an unbiased estimator:
\begin{eqnarray}
\label{shear12PSF2}
&&\frac{1}{2}\frac{\langle\langle (\partial_1f_O)^2-(\partial_2f_O)^2\rangle_g\rangle_{en}}{\langle\langle (\partial_1f_O)^2+(\partial_2f_O)^2+\Delta\rangle_g\rangle_{en}}=-\gamma_1 \\ \nonumber
&&\frac{\langle\langle\partial_1f_O\partial_2f_O\rangle_g\rangle_{en}}{\langle\langle (\partial_1f_O)^2+(\partial_2f_O)^2+\Delta\rangle_g\rangle_{en}}=-\gamma_2
\end{eqnarray}
This is the main result of this paper, and the form of the
unbiased estimator that we propose to use for the actual analysis of
the weak lensing data.

Of course, one could divide the left hand sides of
 eq.(\ref{shear12PSF_wrong}) by $1-\delta_1$ and $1-\delta_2$ to obtain
 an unbiased estimator. This is similar to correcting the measured shear for a multiplicative
bias that is evaluated from the same ensemble of galaxies. In this
sense, eq.(\ref{shear12PSF_wrong}) provides the exact definitions for
the multiplicative biases for $\gamma_1$ and $\gamma_2$. However, since
$\delta_1$ and $\delta_2$ in eq.(\ref{shear12PSF_wrong}) involve many
high order Fourier-space multipole moments of the surface brightness
field, evaluation of these terms from simulations (which are
incomplete anyway) can be highly uncertain. Even worse, 
the multiplicative bias mentioned here is not even a
constant, but depends on the morphological distribution of the
galaxies. This makes the conventional way of measuring shear
correlation functions even more challenging, as one must take into
account the {\it correlations} of the multiplicative 
biases, as will be shown in \S\ref{statistics}. 

In summary, 
according to eq.(\ref{shear12PSF2}), for each shear component, two quantities 
from each galaxy should be kept, and the ratios of their ensemble averages yield unbiased
estimates for the corresponding shear components. Finally, it is
important to note that, to efficiently use eq.(\ref{shear12PSF2}), the
surface brightness of each participating galaxy should be normalized to
have roughly the same maximum value, so that faint galaxies are not much
less weighted than their brighter counterparts. 
The details regarding the optimal weighting scheme as a function of the
galaxy luminosity should also take into account the photon noise. This
is a separate topic, and will be studied in a future work.

\subsection{Comments on Errors due to Finite Number of Galaxies}
\label{error_budget}

Strictly speaking, Eq.(\ref{shear12PSF2}) holds
when we average over an infinite number of galaxies; however, as we shall 
show in this section, the error that we make by having a finite number
of galaxies for averaging is much smaller than the statistical errors,
and thus the estimator remains unbiased for practical applications.

For simplicity,  we use eq.(\ref{q6}) rather than our main equation
[eq.(\ref{shear12PSF2})] in the following discussion, but the conclusion
will be the same for eq.(\ref{shear12PSF2}).

Let us use $\langle X\rangle_N$ to denote the average of the quantity
$X$ over  $N$ galaxies. From eq.(\ref{q3}), we get:
\begin{eqnarray}
\label{q3_error}
\left\langle Q_{11}-Q_{22}\right\rangle_N&=&(1+4\kappa)\left\langle Q_{11}^S-Q_{22}^S\right\rangle_N\\ \nonumber
&+&2\gamma_1\left\langle Q_{11}^S+Q_{22}^S\right\rangle_N\\ \nonumber
\left\langle Q_{12}\right\rangle_N&=&(1+4\kappa)\left\langle Q_{12}^S\right\rangle_N+\gamma_2\left\langle Q_{11}^S+Q_{22}^S\right\rangle_N\\ \nonumber
\left\langle Q_{11}+Q_{22}\right\rangle_N&=&(1+4\kappa)\left\langle Q_{11}^S+Q_{22}^S\right\rangle_N\\ \nonumber
&+&2\gamma_1\left\langle Q_{11}^S-Q_{22}^S\right\rangle_N+4\gamma_2\left\langle Q_{12}^S\right\rangle_N
\end{eqnarray}
Consequently, we have:
\begin{eqnarray}
\label{q4_error}
\frac{1}{2}\frac{\left\langle Q_{11}-Q_{22}\right\rangle_N}{\left\langle Q_{11}+Q_{22}\right\rangle_N}&=&\frac{1}{2}\Delta_1+\gamma_1(1-\Delta_1^2)-\gamma_2\Delta_1\Delta_2\\ \nonumber
\frac{\left\langle Q_{12}\right\rangle_N}{\left\langle Q_{11}+Q_{22}\right\rangle_N}&=&\frac{1}{2}\Delta_2+\gamma_2(1-\Delta_2^2)-\gamma_1\Delta_1\Delta_2
\end{eqnarray}
where
\begin{eqnarray}
\label{def_error}
\Delta_1&=&\frac{\left\langle Q_{11}^S-Q_{22}^S\right\rangle_N}{\left\langle Q_{11}^S+Q_{22}^S\right\rangle_N}\\ \nonumber
\Delta_2&=&\frac{2\left\langle Q_{12}^S\right\rangle_N}{\left\langle Q_{11}^S+Q_{22}^S\right\rangle_N}
\end{eqnarray}
Here, the terms $\Delta_1$ and $\Delta_1\Delta_2$ contribute to
random errors because their ensemble averages vanish, whereas the terms
$\Delta_1^2$ and $\Delta_2^2$ lead to systematic biases
because their ensemble averages do not vanish.
Fortunately, as $\Delta_{1, 2}$ scales as $1/\sqrt{N}$ and $\gamma_{1,
2}\ll 1$, the amplitudes of such systematic biases are always much
smaller than the sizes of the statistical errors. Therefore, the results
from this new type of shear estimators may be regarded as
unbiased for practical applications. Numerical 
verifications will be given in \S\ref{examples}.

\section{Shear Statistics - n-Point Correlations}
\label{statistics}

The cosmic shear field can only be probed statistically. This is mainly
due to the intrinsic variations of the galaxy shapes and the spatial
fluctuations of the shear components. As a result, the
shear statistics is usually studied in the form of n-point spatial
correlation functions of the shear field. The previous discussions and
measurements in the literature are based on ``conventional'' shear
estimators (CSEs), \ie, one often assumes that the following is true:
\begin{equation}
\label{assume_wr}
\langle\Gamma\rangle_{en}=\gamma
\end{equation}
where $\gamma$ can be either $\gamma_1$ or $\gamma_2$, and $\Gamma$ is a CSE for $\gamma$. For individual galaxies, eq.(\ref{assume_wr}) implies:
\begin{equation}
\label{assume_wr2}
\Gamma=\gamma+\Psi
\end{equation}
where $\Psi$ satisfies $\langle\Psi\rangle_{en}=0$\footnote{Note that for our purpose, it is not necessary to know the form of $\Psi$.}. It is usually assumed that $\Psi$'s of different galaxies do not correlate with each other\footnote{This is at least true if the relevant galaxies are separated by a large physical distance. Detailed discussions about the correlations of $\Psi$'s belong to the topic of ``Galaxy Intrinsic Alignment'', which is beyond the scope of this paper.}.  Therefore, the n-point correlation functions of the shear field can be directly measured by the correlations of $\Gamma$'s. 

However, in \S\ref{no_PSF}, we have shown that such a $\Gamma$ at least does not exist in a convenient form. Instead, as proposed in \S\ref{m_zhang08}, we can use the new form of shear estimators defined in eq.(\ref{shear12PSF2}) to probe the cosmic shear in an unbiased way. Let us now find out how to measure the n-point shear correlation functions with the new form of shear estimators. Numerical examples are given in \S\ref{examples}.

For notational convenience, the type of shear measurement in eq.(\ref{shear12PSF2}) can be symbolized as follows:
\begin{equation}
\label{symbol}
\frac{\langle A\rangle_{en}}{\langle B\rangle_{en}}=\gamma
\end{equation}
where $\gamma$ can be either $\gamma_1$ or $\gamma_2$, and $A$ and $B$ are properties of a galaxy, such as those defined in eq.(\ref{shear12PSF2}). Similar to eq.(\ref{assume_wr2}), eq.(\ref{symbol}) implies the following:
\begin{equation}
\label{symbol2}
A=\gamma B+C
\end{equation}
where $C$ satisfies $\langle C\rangle_{en}=0$. If we assume that the $C$ of any galaxy does not correlate with the $B$'s and $C$'s of other galaxies\footnote{Such correlations again belong to the topic of ``Galaxy Intrinsic Alignment,'' which is not considered in this paper.}, the n-point correlation functions of the shear field can be probed using the following relation:
\begin{equation}
\label{symbol3}
\left\langle\gamma(\vec{x}_1)\gamma(\vec{x}_2)\ldots\gamma(\vec{x}_n)\right\rangle_{en}=\frac{\left\langle A(\vec{x}_1)A(\vec{x}_2)\ldots A(\vec{x}_n)\right\rangle_{en}}{\left\langle B(\vec{x}_1)B(\vec{x}_2)\ldots B(\vec{x}_n)\right\rangle_{en}}
\end{equation}   
The ensemble averages are taken over a large number of galaxies whose relative positions $\vec{x}_i-\vec{x}_j$ ($i,j=1,2,\ldots,n$) are fixed. In practice, the n-point shear correlation functions can be measured using:
\begin{equation}
\label{symbol33}
\frac{\sum A(\vec{x}_{i_1})A(\vec{x}_{i_2})\ldots A(\vec{x}_{i_n})}{\sum B(\vec{x}_{i_1})B(\vec{x}_{i_2})\ldots B(\vec{x}_{i_n})}
\end{equation}
where the sum is taken over all the galaxy groups that satisfy the positional constraints. Note that {\it the ratio is taken after the summations}. The standard deviation ($\sigma$) of the correlation function in such a measurement can be calculated as follows:
\begin{eqnarray}
\label{symbol34}
\sigma^2&\doteq& \left\langle\left[\frac{\sum A(\vec{x}_{i_1})A(\vec{x}_{i_2})\ldots A(\vec{x}_{i_n})}{\sum B(\vec{x}_{i_1})B(\vec{x}_{i_2})\ldots B(\vec{x}_{i_n})}\right]^2\right\rangle_{en}\\ \nonumber
&\doteq&\left\langle\frac{\sum A^2(\vec{x}_{i_1})A^2(\vec{x}_{i_2})\ldots A^2(\vec{x}_{i_n})}{\left[\sum B(\vec{x}_{i_1})B(\vec{x}_{i_2})\ldots B(\vec{x}_{i_n})\right]^2}\right\rangle_{en}\\ \nonumber
&=&\frac{N\left\langle A^2(\vec{x}_1)A^2(\vec{x}_2)\ldots A^2(\vec{x}_n)\right\rangle_{en}}{N^2\left\langle B(\vec{x}_1)B(\vec{x}_2)\ldots B(\vec{x}_n)\right\rangle_{en}^2}\\ \nonumber
&=&\frac{1}{N}\left(\frac{\langle A^2\rangle_{en}}{\langle B\rangle_{en}^2}\right)^n
\end{eqnarray} 
where $N$ is the total number of galaxy groups (\eg, the number of galaxy pairs for 2-point correlations) used. 

To summarize, in the new type of shear measurement, the shear correlation function should be measured using the ratio of two ensemble averages, as shown in eq.(\ref{symbol3}). If $B$ in eq.(\ref{symbol}) is viewed as a multiplicative bias, we need to measure the {\it correlations} of these multiplicative biases as well in order to get the correct shear correlation functions.

\section{Numerical Tests} 
\label{examples}

In this section, we show how accurately one can recover the cosmic
shears and their 2-point correlation functions with the method proposed
in \S\ref{alternatives} and \S\ref{statistics}.  For a comparison, we
also show the results using the method of Z08 in the ``conventional''
(but wrong) way, \ie, eq.(\ref{shear12PSF_wrong}), but without taking
into account the biases, $\delta_1$ and $\delta_2$.\footnote{Note that
Z08 did not detect these biases because it did not use
enough galaxies ($10^4$) in the numerical tests. Here, we use $10^7$
galaxies (or pairs of galaxies for the 2-point correlation function) for
each test.} Since our focus is to demonstrate principles, we do not
include photon noise or the pixelation effect in this
paper, and we approximate the PSF as an isotropic Gaussian. 
Including these effects is straightforward.
(Note also that the conventional estimators yield biased results even in this
idealized case.)

 More comprehensive tests of the method of Z08 have been presented in \cite{zhang10b}, which further improves the accuracy of shear measurement by including the second order shear/convergence terms. As shown in that paper, the second order correction is proportional to the convergence $\kappa$. We simply set $\kappa=0$ in the numerical simulations here. The residual systematic error on the recovered shear ($\gamma_1$ or $\gamma_2$) therefore should be expected to have a magnitude comparable to the third order terms in shear (\eg, $\gamma_1^3$, $\gamma_1^2\gamma_2$, $\gamma_1\gamma_2^2$, $\gamma_2^3$).  The errors shown in the rest of this section are all at the $1\sigma$ confidence level.

\subsection{Image Generation}
\label{image}

The mock galaxy images we use in our numerical tests are generated by the algorithms introduced in Z08 and Zhang (2010a), \ie, each galaxy is generated as a collection of point sources.  The reason is simple: one can accurately and easily mimic the lensing effect by displacing the points. It also allows us to generate galaxies of complex morphologies. There are two types of galaxies we use in this paper: 1. randomly oriented regular galaxies, each of which contains an exponential disk in the galactic plane (no bulge); 2. irregular galaxies being made of points generated by the trajectories of 2D random walks. For simplicity, the PSF is always an isotropic Gaussian function, whose scale radius is four times the grid size to avoid the pixelation problem. 
All the lengths in our simulations are in units of the grid size in the rest of this section. The dimension of the grid is $64\times 64$.

\subsection{1-Point Statistics}
\label{one_point}

As our first example, we study how accurately a single input cosmic shear can be recovered by a large number of mock galaxies, \ie, the 1-point statistics. We use the regular type mock galaxies as introduced in \S\ref{image}. Each disk galaxy is composed of ten point sources which are randomly  distributed within a radius of 7. The intensity of a point is an exponentially decaying function of its distance to the center of the disk with a decay length equal to 7. The galactic disk is then projected onto the source plane in a random direction.  For each input shear value, we use $10^7$ mock galaxies to recover the shear. 

To quantify the accuracy of shear recovery, we adopt the standard technique in the weak lensing community by using the ``multiplicative bias'' $m_i$ and the ``additive bias'' $c_i$, which are defined as:
\begin{eqnarray}
\label{mc1}
&&\gamma_1^{measured}=(1+m_1)\gamma_1^{input}+c_1\\ \nonumber
&&\gamma_2^{measured}=(1+m_2)\gamma_2^{input}+c_2
\end{eqnarray}

Our simulations use six sets of input shear values ($\gamma_1$,
$\gamma_2$). They are: ($0.05$, $-0.05$), ($0.03$, $-0.03$), ($0.01$,
$-0.01$), ($-0.01$, $0.01$), ($-0.03$, $0.03$), ($-0.05$, $0.05$). The
recovered shear values as well as the linear fitting results for
$\gamma_1$ and $\gamma_2$ are shown in Table \ref{one_p_gamma1} and
Table \ref{one_p_gamma2} respectively.  
Note that we also list the values of $\chi^2$'s and $Q$'s for the
goodness of linear fitting (see \citealt{press92} for details). For a
comparison, we show in the last column of each table the quality of the
shear recoveries using the ``conventional'' (but wrong) way, given by
eq.(\ref{shear12PSF_wrong}).  

The tables show that the shear recovery can be very accurate if we use
the method of Z08 in the proper way, \ie, eq.(\ref{shear12PSF2}). On the
other hand, if Z08 is used in the ``conventional'' way, \ie,
eq.(\ref{shear12PSF_wrong}), we find nonzero multiplicative biases
($m_1$, $m_2$) for both $\gamma_1$ and $\gamma_2$. The additive biases
$c_1$ and $c_2$ are consistent with being zero. These results agree with
our conclusions in \S\ref{m_zhang08}.    

\begin{table}

	\centering 
	\begin{tabular}
		{c|c|c}
		
		Input $\gamma_1$ & $\gamma_1$ measured in  & $\gamma_1$ measured in \\
		  & the proper way & the ``conventional'' way \\

		\hline \hline  0.05  & 0.05004$\pm$0.00008 & 0.05217$\pm$0.00009\\

		\hline 0.03 & 0.03004$\pm$0.00008 & 0.03132$\pm$0.00009\\

		\hline 0.01 & 0.01006$\pm$0.00008 & 0.01049$\pm$0.00009\\

		\hline $-0.01$ & $-0.01002\pm$0.00008 & $-0.01046\pm$0.00009\\

		\hline $-0.03$ & $-0.02997\pm$0.00008& $-0.03126\pm$0.00009\\

		\hline $-0.05$ & $-0.05005\pm$0.00008 & $-0.05218\pm$0.00009\\

		\hline \hline & &\\
		 & & \\

		Linear Fitting & Using $\gamma_1$'s from  & Using $\gamma_1$'s from \\
		Results  & the proper way & the ``conventional'' way \\

		\hline\hline $m_1$ & $(0.7\pm1.0)\times10^{-3}$ & $(43.4\pm1.1)\times10^{-3}$\\
		
		\hline $c_1$ & $(1.6\pm3.4)\times 10^{-5}$ & $(1.2\pm3.7)\times10^{-5}$\\
		
		\hline ($\chi^2_1$,  $Q_1$) & (0.58, 0.96) & (0.58, 0.96)\\
		
		\hline\hline
		
\end{tabular}
\caption{In the middle and right columns of the upper part of the table,
 we list the measured $\gamma_1$'s from the proper way and the
 ``conventional'' way of using the method of Z08, respectively. The
 lower part of the table shows the multiplicative biases, the additive
 biases, and the goodness of the linear fittings for both cases. The
 definition of the linear fitting here is given in eq.(\ref{mc1}).} 
\label{one_p_gamma1}
\end{table}

\begin{table}

	\centering 
	\begin{tabular}
		{c|c|c}
		
		Input $\gamma_2$ & $\gamma_2$ measured in  & $\gamma_2$ measured in \\
		  & the proper way & the ``conventional'' way \\

		\hline \hline  0.05 & 0.05011$\pm$0.00008 & 0.05226$\pm$0.00009\\

		\hline 0.03 & 0.02996$\pm$0.00008 & 0.03126$\pm$0.00009\\

		\hline 0.01 & 0.00989$\pm$0.00008 & 0.01033$\pm$0.00009 \\

		\hline $-0.01$ & $-0.01010\pm$0.00008 & $-0.01055\pm$0.00009\\

		\hline $-0.03$ & $-0.03000\pm$0.00008 & $-0.03126\pm$0.00009\\

		\hline $-0.05$ & $-0.05011\pm$0.00008 & $-0.05225\pm$0.00009\\

		\hline \hline & &\\
		 & & \\

		Linear Fitting & Using $\gamma_2$'s from  & Using $\gamma_2$'s from \\
		Results  & the proper way & the ``conventional'' way \\

		\hline\hline $m_2$ & $(1.4\pm1.0)\times10^{-3}$ & $(44.2\pm1.1)\times10^{-3}$\\
		
		\hline $c_2$ & $(-4.1\pm3.4)\times 10^{-5}$ & $(-3.5\pm3.7)\times10^{-5}$\\
		
		\hline ($\chi^2_2$, $Q_2$) & (3.3, 0.51) & (3.4, 0.50)\\
		
		\hline\hline

\end{tabular}
\caption{Same as table \ref{one_p_gamma1}, except that it is for $\gamma_2$.}
\label{one_p_gamma2}	
\end{table}

\subsection{2-Point Correlations}
\label{two_point}

Let us now test the recovery of the 2-point shear correlations using the
method of Z08 in the way that is proposed in \S\ref{statistics}.  In
each test, we use two large groups of galaxies that are 
given ($\gamma_1$, $\gamma_2$) and ($\gamma_1'$, $\gamma_2'$),
respectively. Each group of each test contains $10^7$ galaxies. The
values of $\gamma_1$, $\gamma_2$, $\gamma_1'$, and $\gamma_2'$ vary from
galaxy to galaxy. They are assumed to be normally distributed with the
following covariance matrix:
$\langle\gamma_1^2\rangle=\langle\gamma_1'^2\rangle=\langle\gamma_2^2\rangle=\langle\gamma_2'^2\rangle=0.04^2$,
$\langle\gamma_1\gamma_2\rangle=\langle\gamma_1\gamma_2'\rangle=\langle\gamma_1'\gamma_2\rangle=\langle\gamma_1'\gamma_2'\rangle=0$,
and $\langle\gamma_1\gamma_1'\rangle$ and
$\langle\gamma_2\gamma_2'\rangle$ are to be specified in each test. The
purpose of the test is to find out how accurately the 2-point correlations
$\langle\gamma_1\gamma_1'\rangle$ and $\langle\gamma_2\gamma_2'\rangle$
can be recovered. For the tests presented here, we use the irregular type of mock galaxies, each of which is made of ten point sources generated by the 2D random walks. Each step size of the random walks is a random number between 0 and 2. The radius of each galaxy is limited to be less than 7. 

In table \ref{two_p_gamma1} and \ref{two_p_gamma2}, we report the results of six tests with six different sets of ($\langle\gamma_1\gamma_1'\rangle$, $\langle\gamma_2\gamma_2'\rangle$): (0.001,0.001), (0.0006, 0.0006), (0.0002, 0.0002), ($-0.0002$, $-0.0002$), ($-0.0006$, $-0.0006$), ($-0.001$, $-0.001$). As in \S\ref{one_point},  to characterize the accuracy of the method, we again use the multiplicative biases ($m_{11}$, $m_{22}$) and additive biases ($c_{11}$, $c_{22}$) that are defined as follows:
\begin{eqnarray}
\label{mc2}
&&\langle\gamma_1\gamma_1'\rangle^{measured}=(1+m_{11})\langle\gamma_1\gamma_1'\rangle^{input}+c_{11}\\ \nonumber
&&\langle\gamma_2\gamma_2'\rangle^{measured}=(1+m_{22})\langle\gamma_2\gamma_2'\rangle^{input}+c_{22}
\end{eqnarray}
We also show the results from the ``conventional'' way of using Z08 for a comparison. 

Our results again show no systematic errors for the proper way of using
Z08. In contrast, the ``conventional'' way tends to underestimate the
amplitudes of the shear correlations (negative multiplicative bias).
The signs of the multiplicative biases in the ``conventional'' cases are
opposite to those found in \S\ref{one_point}. This is because we have
used two different types of mock galaxies.

\begin{table}

	\centering 
	\begin{tabular}
		{c|c|c}
		
		Input $\langle\gamma_1\gamma_1'\rangle$ & $\langle\gamma_1\gamma_1'\rangle$ measured in  & $\langle\gamma_1\gamma_1'\rangle$ measured in \\
		  & the proper way & the ``conventional'' way \\

		\hline \hline  $10^{-3}$  & $(0.99\pm0.03)\times10^{-3}$ & $(0.71\pm0.02)\times10^{-3}$\\

		\hline $6\times10^{-4}$ & $(6.3\pm0.3)\times10^{-4}$ & $(4.5\pm0.2)\times10^{-4}$\\

		\hline $2\times10^{-4}$ & $(1.7\pm0.3)\times10^{-4}$ & $(1.2\pm0.2)\times10^{-4}$\\

		\hline $-2\times10^{-4}$ & $(-2.1\pm0.3)\times10^{-4}$ & $(-1.4\pm0.2)\times10^{-4}$\\

		\hline $-6\times10^{-4}$ & $(-6.1\pm0.3)\times10^{-4}$ & $(-4.5\pm0.2)\times10^{-4}$\\

		\hline $-10^{-3}$ & $(-0.99\pm0.03)\times10^{-3}$ & $(-0.72\pm0.02)\times10^{-3}$\\

		\hline \hline & &\\
		 & & \\

		Linear Fitting & Using $\langle\gamma_1\gamma_1'\rangle$'s from  & Using $\langle\gamma_1\gamma_1'\rangle$'s from \\
		Results  & the proper way & the ``conventional'' way \\

		\hline\hline $m_{11}$ & -0.001$\pm$0.018 & -0.281$\pm$0.012\\
		
		\hline $c_{11}$ & $(-0.18\pm1.2)\times 10^{-5}$ & $(-4.1\pm7.9)\times10^{-6}$\\
		
		\hline ($\chi^2_{11}$, $Q_{11}$) & (2.5, 0.64) & (2.3, 0.68)\\
		
		\hline\hline

\end{tabular}
\caption{In the middle and right columns of the upper part of the table,
 we list the measured $\langle\gamma_1\gamma_1'\rangle$'s from the
 proper way and the ``conventional'' way of using the method of Z08, respectively. The lower part of the table shows the multiplicative biases, the additive biases, and the goodness of the linear fittings for both cases. The definition of linear fitting here is given in eq.(\ref{mc2}).}
\label{two_p_gamma1}
\end{table}

\begin{table}

	\centering 
	\begin{tabular}
		{c|c|c}
		
		Input $\langle\gamma_2\gamma_2'\rangle$ & $\langle\gamma_2\gamma_2'\rangle$ measured in  & $\langle\gamma_2\gamma_2'\rangle$ measured in \\
		  & the proper way & the ``conventional'' way \\

		\hline \hline  $10^{-3}$  & $(0.99\pm0.03)\times10^{-3}$ & $(0.71\pm0.02)\times10^{-3}$\\

		\hline $6\times10^{-4}$ & $(5.8\pm0.3)\times10^{-4}$ & $(4.2\pm0.2)\times10^{-4}$\\

		\hline $2\times10^{-4}$ & $(2.0\pm0.3)\times10^{-4}$ & $(1.3\pm0.2)\times10^{-4}$\\

		\hline $-2\times10^{-4}$ & $(-1.7\pm0.3)\times10^{-4}$ & $(-1.3\pm0.2)\times10^{-4}$\\

		\hline $-6\times10^{-4}$ & $(-6.0\pm0.3)\times10^{-4}$ & $(-4.4\pm0.2)\times10^{-4}$\\

		\hline $-10^{-3}$ & $(-1.03\pm0.03)\times10^{-3}$ & $(-0.73\pm0.02)\times10^{-3}$\\

		\hline \hline & &\\
		 & & \\

		Linear Fitting & Using $\langle\gamma_2\gamma_2'\rangle$'s from  & Using $\langle\gamma_2\gamma_2'\rangle$'s from \\
		Results  & the proper way & the ``conventional'' way \\

		\hline\hline $m_{22}$ & 0.000$\pm$0.018 & -0.284$\pm$0.012\\
		
		\hline $c_{22}$ & $(-0.61\pm1.2)\times 10^{-5}$ & $(-6.2\pm7.9)\times10^{-6}$\\
		
		\hline ($\chi^2_{22}$, $Q_{22}$) & (1.8, 0.77) & (1.6, 0.82)\\
		
		\hline\hline
		
\end{tabular}
\caption{Same as table \ref{two_p_gamma1}, except that it is for $\langle\gamma_2\gamma_2'\rangle$.}
\label{two_p_gamma2}
\end{table}

\section{Summary}
\label{summary}

Conventionally, in the studies of weak lensing, for each shear component
($\gamma_1$ or $\gamma_2$), one hopes to construct a single
quantity from each background galaxy image, whose ensemble average is
equal to the true  value of a component of the true shear field. We have
shown that such conventional shear estimators (CSEs) {\it do not exist
in convenient forms} even in the absence of the PSF.

Based on the method of Zhang (2008), we have proposed to measure the cosmic
shear in a new way: using the ratio of the ensemble averages of two
galaxy properties to estimate each shear component. (Also
see \S~9.2 of Weinberg (2008) for a similar study.)
We have shown that, using
both analytic analyses and numerical examples, the new way of
estimating cosmic shears is unbiased, and does not contain
systematic errors to the first 
order in shear at least. The new type of shear measurement demands shear
statistics such as n-point correlation functions to be measured in an unconventional way as well, but with little additional cost. 

\section*{Acknowledgements}

The author acknowledges Gary Bernstein, Toshifumi Futamase, Yi Mao, Pengjie Zhang for helpful discussions, and Shanghai Astronomical Observatory (SHAO), National Astronomical Observatories of China (NAOC) for their hospitality. In particular, JZ would like to thank Gary Bernstein for pointing out a mistake in the proof for the nonexistence of CSEs in the presence of PSF in the previous version of this manuscript. JZ is currently supported by the TCC Fellowship of Texas Cosmology Center of the University of Texas at Austin, and was previously supported by the TAC Fellowship of the Theoretical Astrophysics Center of UC Berkeley, where part of this work was done. EK is supported in part by NSF grants AST-0807649 and PHY-0758153 and NASA grant NNX08AL43G.

\bibliographystyle{mn2e}

\vskip 1cm

\section*{Appendix A -- The Derivatives of the Spin-2 CSEs}
\label{appendixA}

Under a clockwise coordinate rotation of angle $\theta$, the cosmic shear components transform according to the following rule:
\begin{equation}
\label{shear_rotation}
\gamma_1^{\theta}+\ima\gamma_2^{\theta}=\left(\gamma_1+\ima\gamma_2\right)\exp(\ima 2\theta)
\end{equation}
Using the chain rule, we then find:
\begin{equation}
\label{shear_deriv_rotation}
\frac{\partial}{\partial\gamma_1^{\theta}}+\ima\frac{\partial}{\partial\gamma_2^{\theta}}=\left(\frac{\partial}{\partial\gamma_1}+\ima\frac{\partial}{\partial\gamma_2}\right)\exp(\ima 2\theta)
\end{equation}
On the other hand, we know that the spin-2 shear estimators $(\Gamma_1, \Gamma_2)$ transform as:
\begin{equation}
\label{Estimator_transform}
\Gamma_1^{\theta}+\ima\Gamma_2^{\theta}=\left(\Gamma_1+\ima\Gamma_2\right)\exp(\ima 2\theta)
\end{equation}
Apply eq.(\ref{shear_deriv_rotation}) onto eq. (\ref{Estimator_transform}), we get:
\begin{eqnarray}
\label{spin-4}
&&\left(\frac{\partial\Gamma_1^{\theta}}{\partial\gamma_1^{\theta}}-\frac{\partial\Gamma_2^{\theta}}{\partial\gamma_2^{\theta}}\right)+\ima\left(\frac{\partial\Gamma_1^{\theta}}{\partial\gamma_2^{\theta}}+\frac{\partial\Gamma_2^{\theta}}{\partial\gamma_1^{\theta}}\right)\\ \nonumber
&=&\left[\left(\frac{\partial\Gamma_1}{\partial\gamma_1}-\frac{\partial\Gamma_2}{\partial\gamma_2}\right)+\ima\left(\frac{\partial\Gamma_1}{\partial\gamma_2}+\frac{\partial\Gamma_2}{\partial\gamma_1}\right)\right]\exp(\ima 4\theta)
\end{eqnarray}
Eq.(\ref{spin-4}) clearly shows that the derivatives $(\partial_{\gamma_1}\Gamma_1-\partial_{\gamma_2}\Gamma_2)+\ima(\partial_{\gamma_2}\Gamma_1+\partial_{\gamma_1}\Gamma_2)$ form a spin-4 quantity. We can also apply eq.(\ref{shear_deriv_rotation}) onto the complex conjugate of eq. (\ref{Estimator_transform}), and get:
\begin{eqnarray}
\label{spin-0}
&&\left(\frac{\partial\Gamma_1^{\theta}}{\partial\gamma_1^{\theta}}+\frac{\partial\Gamma_2^{\theta}}{\partial\gamma_2^{\theta}}\right)+\ima\left(\frac{\partial\Gamma_1^{\theta}}{\partial\gamma_2^{\theta}}-\frac{\partial\Gamma_2^{\theta}}{\partial\gamma_1^{\theta}}\right)\\ \nonumber
&=&\left(\frac{\partial\Gamma_1}{\partial\gamma_1}+\frac{\partial\Gamma_2}{\partial\gamma_2}\right)+\ima\left(\frac{\partial\Gamma_1}{\partial\gamma_2}-\frac{\partial\Gamma_2}{\partial\gamma_1}\right)
\end{eqnarray}
Therefore, $\partial_{\gamma_1}\Gamma_1+\partial_{\gamma_2}\Gamma_2$ is a scalar, and $\partial_{\gamma_2}\Gamma_1-\partial_{\gamma_1}\Gamma_2$ is a pseudo-scalar (it has odd parity). 

\section*{Appendix B -- Some Mathematical Details Regarding Galaxies of Elliptical Shapes}
\label{appendixB}

The surface brightness profiles of the galaxies used in \S\ref{spin_2_esti} are parametrized as $f_S(R)$, where $R=a(x^2+y^2)+b(x^2-y^2)+2cxy$, and $x$ and $y$ are the coordinate variables. Curves of constant values of $R$ can be ellipses, hyperbolas, parabolas, or lines depending on the values of $a$, $b$, and $c$. For our purposes, we only need ellipses. This requires $a+b>0$ and $a^2-b^2>c^2$ according to the matrix theory. Since linear coordinate transformations do not spoil these relations, weakly lensed images of these galaxies are still ellipses. For the same reason, one can also easily show that the lensed images are still parametrized by the same function $f_S(R)$ with $R=a'(x^2+y^2)+b'(x^2-y^2)+2c'xy$, where $a'=a(1-2\kappa)-2\gamma_1b-2\gamma_2c$, $b'=b(1-2\kappa)-2\gamma_1a$, and $c'=c(1-2\kappa)-2\gamma_2a$.

Let us now show that the ellipticity parameters $(\epsilon_1, \epsilon_2)$ for galaxies of the form  $f_S\left[a(x^2+y^2)+b(x^2-y^2)+2cxy\right]$ are equal to $(-b/a, -c/a)$. According to the definitions in \S\ref{classic}, we have:
\begin{equation}
\label{e1e2}
Q_{ij}=\int d^2\vec{x}x_ix_jf_S\left[a(x^2+y^2)+b(x^2-y^2)+2cxy\right]
\end{equation}
This integration can be carried out using the following linear coordinate transformation:
\begin{eqnarray}
\label{ct_e1e2}
&&x=\frac{1}{\sqrt{a+b}}(\beta x'+\alpha y')\\ \nonumber
&&y=\frac{1}{\sqrt{a-b}}(\alpha x'+\beta y')
\end{eqnarray} 
where 
\begin{eqnarray}
\label{ab_ct}
\alpha&=&\frac{1}{2}\left(\sqrt{1-\frac{c}{\sqrt{a^2-b^2}}}-\sqrt{1+\frac{c}{\sqrt{a^2-b^2}}}\right)\\ \nonumber
&\times&\sqrt{\frac{a^2-b^2}{a^2-b^2-c^2}}\\ \nonumber
\beta&=&\frac{1}{2}\left(\sqrt{1-\frac{c}{\sqrt{a^2-b^2}}}+\sqrt{1+\frac{c}{\sqrt{a^2-b^2}}}\right)\\ \nonumber
&\times&\sqrt{\frac{a^2-b^2}{a^2-b^2-c^2}}
\end{eqnarray}
Using $(x', y')$ defined in eq.(\ref{ct_e1e2}) to replace $(x, y)$ in eq.(\ref{e1e2}), we get:
\begin{eqnarray}
\label{QQs}
&&Q_{11}-Q_{22}\\ \nonumber
&=&\frac{\beta^2-\alpha^2}{\sqrt{a^2-b^2}}\int d^2\vec{x'}\left[\frac{a(\beta^2-\alpha^2)}{a^2-b^2}(x'^2-y'^2)\right.\\ \nonumber
&-&\left.\frac{b(\alpha^2+\beta^2)}{a^2-b^2}(x'^2+y'^2)-\frac{4b\alpha\beta}{a^2-b^2}x'y'\right]f_S(x'^2+y'^2)\\ \nonumber
&&Q_{11}+Q_{22}\\ \nonumber
&=&\frac{\beta^2-\alpha^2}{\sqrt{a^2-b^2}}\int d^2\vec{x'}\left[-\frac{b(\beta^2-\alpha^2)}{a^2-b^2}(x'^2-y'^2)\right.\\ \nonumber
&+&\left.\frac{a(\alpha^2+\beta^2)}{a^2-b^2}(x'^2+y'^2)+\frac{4a\alpha\beta}{a^2-b^2}x'y'\right]f_S(x'^2+y'^2)\\ \nonumber
&&2Q_{12}\\ \nonumber
&=&\frac{\beta^2-\alpha^2}{\sqrt{a^2-b^2}}\int d^2\vec{x'}\left[\frac{2\alpha\beta}{\sqrt{a^2-b^2}}(x'^2+y'^2)\right.\\ \nonumber
&+&\left.\frac{2(\alpha^2+\beta^2)}{\sqrt{a^2-b^2}}x'y'\right]f_S(x'^2+y'^2)
\end{eqnarray} 
Given that the function $f_S$ only depends on $x'^2+y'^2$, we have:
\begin{eqnarray}
\label{facts}
&&\int d^2\vec{x'}(x'^2-y'^2)f_S(x'^2+y'^2)\\ \nonumber
&=&\int d^2\vec{x'}x'y'f_S(x'^2+y'^2)=0
\end{eqnarray}
It is now straightforward to calculate $(\epsilon_1, \epsilon_2)$:
\begin{eqnarray}
\label{e_12s}
&&\epsilon_1=\frac{Q_{11}-Q_{22}}{Q_{11}+Q_{22}}=-\frac{b}{a}\\ \nonumber
&&\epsilon_2=\frac{2Q_{12}}{Q_{11}+Q_{22}}=-\frac{c}{a}
\end{eqnarray}

\section*{Appendix C -- Why CSEs Do Not Likely Exist In the Presence of PSF }
\label{appendixC}

For technical convenience, let us work in Fourier space. According to the notations in \S\ref{classic}, we use $\widetilde{f_L}(\vec{k}^L)$ and $\widetilde{f_S}(\vec{k}^S)$ to denote the Fourier transformations of the lensed galaxy image $f_L(\vtl)$ and the original galaxy image $f_S(\vts)$ respectively. Their relations are given by the following equations:
\begin{eqnarray}
\label{Fourier}
&&\widetilde{f_L}(\vec{k}^L)=\int d^2\vtl e^{\ima\vec{k}^L\cdot\vtl}f_L(\vtl)\\ \nonumber
&&\widetilde{f_S}(\vec{k}^S)=\int d^2\vts e^{\ima\vec{k}^S\cdot\vts}f_S(\vts)
\end{eqnarray}
Using the relations defined in eq.(\ref{fifstits}), we get:
\begin{eqnarray}
\label{Fourier2}
\widetilde{f_L}(\vec{k}^L)&=&\int d^2\vts \left\vert \mathrm{det} \left(\frac{\partial \vtl}{\partial \vts}\right)\right\vert e^{\ima\vec{k}^L\cdot(\mathbf{A}\vts)}f_S(\vts)\\ \nonumber
&=&\vert \mathrm{det} (\mathbf{A})\vert\int d^2\vts e^{\ima(\mathbf{A}\vec{k}^L)\cdot\vts}f_S(\vts)\\ \nonumber
&=&\vert \mathrm{det} (\mathbf{A})\vert\widetilde{f_S}(\mathbf{A}\vec{k}^L)
\end{eqnarray}
Eq. (\ref{Fourier2}) simply means that under lensing,  the Fourier transformation of the galaxy image is changed from  $\widetilde{f_S}(\vec{k})$ to $\vert \mathrm{det} (\mathbf{A})\vert\widetilde{f_S}(\mathbf{A}\vec{k})$.

Due to the presence of the PSF, the Fourier transformation of the observed image $\widetilde{f_O}$ is related to that of the lensed image $\widetilde{f_L}$ via:
\begin{equation}
\label{f_O}
\widetilde{f_O}(\vec{k})=\widetilde{W}(\vec{k})\widetilde{f_L}(\vec{k})
\end{equation}
where $\widetilde{W}(\vec{k})$ is the Fourier transformation of the PSF. Without loss of generality, in the rest of our discussion, we use the isotropic Gaussian PSF, \ie, $\widetilde{W}(\vec{k})=\widetilde{W}_{\beta}(\vec{k})=\exp(-\beta^2\left\vert\vec{k}\right\vert^2/2)$. The advantage of working in Fourier space is that the PSF is included as a multiplicative factor, rather than a convolution as in real space.  Combining eq.(\ref{Fourier2}) and eq.(\ref{f_O}), we get:
\begin{equation}
\label{f_O2}
\widetilde{f_O}(\vec{k})=\widetilde{W}_{\beta}(\vec{k})\vert \mathrm{det} (\mathbf{A})\vert\widetilde{f_S}(\mathbf{A}\vec{k})
\end{equation}
Since the PSF profile in Fourier space typically falls off quickly when the wave number exceeds the inverse of the size of the PSF, it strongly suppresses the power of the observed images on small scales. Therefore, only a finite number of Fourier modes are available for providing shape information. Recovering information on arbitrarily small scales is never feasible in practice due to noise and numerical problems.  

To form CSEs, let us use the multipole moments of galaxy images to represent the shape information, which are defined as:
\begin{equation}
\label{defineM}
M_{ij}=\int d^2\vec{k}k_1^ik_2^j\widetilde{f_O}(\vec{k})
\end{equation}
where $i$ and $j$ are non-negative integers. Note that due to the finite degrees of freedom of the shape information, one can equivalently choose other basis (\eg, shapelets) to study the same issue without affecting the conclusion. For simplicity but without loss of generality, let us only consider galaxies that are invariant under the parity transformation $\vec{x}\to -\vec{x}$. Note that weak lensing does not change this property. For this type of galaxies, the imaginary part of $\widetilde{f_O}(\vec{k})$ is always zero, and $M_{ij}$'s are real. Furthermore, $M_{ij}$ is zero when $i+j$ is an odd number. In this case, the shear estimators $\Gamma_1$ and $\Gamma_2$ can be written as functions of the $M_{ij}$ with $i+j$ being even numbers only. For some of the lowest order $M_{ij}$'s, we can find out how they transform under lensing using eq.(\ref{f_O2}):
\begin{eqnarray}
\label{Mldl}
M_{ij}&=&\int d^2\vec{k}k_1^ik_2^j\widetilde{W}_{\beta}(\vec{k})\vert \mathrm{det} (\mathbf{A})\vert\widetilde{f_S}(\mathbf{A}\vec{k})\\ \nonumber
&=&\int d^2\vec{k}(\mathbf{A}^{-1}\vec{k})_1^i(\mathbf{A}^{-1}\vec{k})_2^j\widetilde{W}_{\beta}(\mathbf{A}^{-1}\vec{k})\widetilde{f_S}(\vec{k})
\end{eqnarray}
The last step in the above equation is achieved by redefining $\mathbf{A}\vec{k}$ as $\vec{k}$. By keeping the terms in matrix $\mathbf{A}$ up to the first order in shear, one can straightforwardly show the following:
\begin{eqnarray}
\label{Mij_trans}
&&M_{20}-M_{02}\\ \nonumber
&=&(1-2\kappa)(M_{20}^S-M_{02}^S)+\kappa\beta^2\left(M_{40}^S-M_{04}^S\right)\\ \nonumber
&+&2\gamma_2\beta^2\left(M_{31}^S-M_{13}^S\right)\\ \nonumber
&-&\gamma_1\left[2(M_{20}^S+M_{02}^S)-\beta^2\left(M_{40}^S+M_{04}^S-2M_{22}^S\right)\right]\\ \nonumber
\\ \nonumber
&&M_{11}\\ \nonumber
&=&(1-2\kappa)M_{11}^S-\gamma_2\left(M_{20}^S+M_{02}^S-2\beta^2 M_{22}^S\right)\\ \nonumber
&+&\kappa\beta^2\left(M_{31}^S+M_{13}^S\right)+\gamma_1\beta^2\left(M_{31}^S-M_{13}^S\right)\\ \nonumber
\\ \nonumber
&&M_{20}+M_{02}\\ \nonumber
&=&(1-2\kappa)(M_{20}^S+M_{02}^S)+\kappa\beta^2\left(M_{40}^S+2M_{22}^S+M_{04}^S\right)\\ \nonumber
&-&\gamma_1\left[2(M_{20}^S-M_{02}^S)-\beta^2\left(M_{40}^S-M_{04}^S\right)\right]\\ \nonumber
&-&\gamma_2\left[4M_{11}^S-2\beta^2\left(M_{31}^S+M_{13}^S\right)\right]
\end{eqnarray}
where 
\begin{equation}
\label{MS}
M_{ij}^S=\int d^2\vec{k}k_1^ik_2^j\widetilde{W}_{\beta}(\vec{k})\widetilde{f_S}(\vec{k})
\end{equation}

The conventional cosmic shear estimators ($\Gamma_1, \Gamma_2$) are functions of $M_{ij}$. According to our discussion in \S\ref{spin_2_esti}, the two functions have to satisfy the following relation:
\begin{equation}
\label{requirement_again}
\frac{\partial\Gamma_1}{\partial\gamma_1}+\frac{\partial\Gamma_2}{\partial\gamma_2}=2
\end{equation} 
In the presence of PSF, we are unable to find out whether one can find CSEs that satisfy the most general requirement given in eq.(\ref{requirement_again}). However, we can show that there do not exist CSEs satisfying a slightly stronger condition:
\begin{equation}
\label{requirement_stronger}
\frac{\partial\Gamma_1}{\partial\gamma_1}=\frac{\partial\Gamma_2}{\partial\gamma_2}=1
\end{equation} 
In addition to eq. (\ref{requirement_again}), eq.(\ref{requirement_stronger}) simply imposes another requirement that the spin-4 parts of the derivatives of the shear estimators with respect to the shears are zero.  
We can rewrite eq.(\ref{requirement_stronger}) for, \eg, $\Gamma_1$, using the chain rule as:
\begin{equation}
\label{Gamma_1_chain_rule}
1=\sum_{ij}\frac{\partial\Gamma_1}{\partial M_{ij}}\frac{\partial M_{ij}}{\partial \gamma_1}
\end{equation}
Since there are only a finite number of multipole moments available for constructing the shear estimators, we assume the maximum value of $i+j$ is $N$ (N is an even integer), \ie, we have:
 \begin{equation}
\label{Gamma_1_chain_rule2}
1=\sum_{i+j\le N}\frac{\partial\Gamma_1}{\partial M_{ij}}\frac{\partial M_{ij}}{\partial \gamma_1}
\end{equation}
Because the right side of eq.(\ref{Gamma_1_chain_rule2}) is evaluated at $\gamma_1=\gamma_2=\kappa=0$, both $\partial\Gamma_1/\partial M_{ij}$ and $\partial M_{ij}/\partial \gamma_1$ are functions of only $M_{ij}^S$'s, \ie, the multipole moments in the absence of lensing. 

As in eq.(\ref{Mij_trans}), one can show in general that $\partial M_{ij}/\partial \gamma_1$ involves higher order multipole moments, \ie,  $M_{i'j'}^S$'s with $i'+j'>i+j$. In particular, when $i+j=N$, $\partial M_{ij}/\partial \gamma_1$ depends linearly on $M_{i'j'}^S$'s with $i'+j'>N$. To satisfy the constraint in eq.(\ref{Gamma_1_chain_rule2}), the coefficients in front of the terms proportional to $M_{i'j'}^S$'s ($i'+j'>N$) must vanish because they are independent of the multipole moments with $i+j\le N$.  As a result, we find that $\partial\Gamma_1/\partial M_{ij}$ has to vanish when $i+j=N$. In other words, we have:
 \begin{equation}
\label{Gamma_1_chain_rule3}
1=\sum_{i+j\le N-2}\frac{\partial\Gamma_1}{\partial M_{ij}}\frac{\partial M_{ij}}{\partial \gamma_1}
\end{equation}
We can now recursively use the above reasoning to show that there does
not exist an $N$ that can satisfy
eq.(\ref{Gamma_1_chain_rule2}). Therefore, we can never find CSEs of the
type defined in eq.(\ref{requirement_stronger}). Though, one may still
expect to find CSEs based on the most general requirement defined in
eq.(\ref{requirement_again}). However, even in this case, CSEs {\it must} reduce to highly nonlinear forms for galaxies of pure elliptical shapes when the PSF effect is small, as we have shown in \S\ref{spin_2_esti}. This feature is already sufficient for arguing against their usefulness in practice.

\section*{Appendix D -- Derivation of the Multiplicative Biases Resulting From a Misuse of Z08}
\label{appendixD}

Let us now calculate the terms $\delta_1$ and $\delta_2$ defined in eq.(\ref{shear12PSF_wrong}). The averages of the spatial derivatives of the surface brightness field of a single galaxy can be related to the Fourier modes of the image. The Fourier transformation has been defined in eq.(\ref{Fourier}) in Appendix C. Following the notations of Appendix C, we find:
\begin{equation}
\label{RF_derive}
\partial_if_O(\vec{x})=\int\frac{d^2\vec{k}}{(2\pi)^2}(-\ima k_i)e^{-\ima\vec{k}\cdot\vec{x}}\widetilde{f_O}(\vec{k})
\end{equation}
and
\begin{eqnarray}
\label{RF_derivative_aves}
&&\left\langle\partial_if_O(\vec{x})\partial_jf_O(\vec{x})\right\rangle_g\\ \nonumber
&&=\frac{\int_Sd^2\vec{x}}{S}\partial_if_O(\vec{x})\partial_jf_O(\vec{x})\\ \nonumber
&&=\frac{\int_Sd^2\vec{x}}{S}\frac{\int d^2\vec{k}\int d^2\vec{k'}}{(2\pi)^4}(-k_ik'_j)e^{-\ima(\vec{k}+\vec{k'})\cdot\vec{x}}\widetilde{f_O}(\vec{k})\widetilde{f_O}(\vec{k'})\\ \nonumber
&&=\frac{1}{S}\frac{\int d^2\vec{k}\int d^2\vec{k'}}{(2\pi)^4}(-k_ik'_j)(2\pi)^2\delta_D^2(\vec{k}+\vec{k'})\widetilde{f_O}(\vec{k})\widetilde{f_O}(\vec{k'})\\ \nonumber
&&=\frac{1}{S}\frac{\int d^2\vec{k}}{(2\pi)^2}k_ik_j\left\vert\widetilde{f_O}(\vec{k})\right\vert^2
\end{eqnarray}
where $S$ is the total area of the map containing the galaxy. Similarly, one can derive the following relation:
\begin{eqnarray}
\label{RF_derive2}
\left\langle\vec{\nabla}f\cdot\vec{\nabla}(\nabla^2f_O)\right\rangle_g=-\frac{1}{S}\frac{\int d^2\vec{k}}{(2\pi)^2}\left\vert\vec{k}\right\vert^4\left\vert\widetilde{f_O}(\vec{k})\right\vert^2
\end{eqnarray} 
Eq.(\ref{RF_derivative_aves}) and eq.(\ref{RF_derive2}) allow us to transform eq.(\ref{shear12PSF_wrong}) into its version in Fourier space:
\begin{eqnarray}
\label{shearFourier}
&&\frac{1}{2}\left\langle\frac{P_{20}-P_{02}}{P_{20}+P_{02}-\beta^2D_4/2}\right\rangle_{en}=-\gamma_1(1-\delta_1)\\ \nonumber
&&\left\langle\frac{P_{11}}{P_{20}+P_{02}-\beta^2D_4/2}\right\rangle_{en}=-\gamma_2(1-\delta_2)
\end{eqnarray}
where 
\begin{eqnarray}
\label{defineP}
&&P_{ij}=\int d^2\vec{k}k_1^ik_2^j\left\vert\widetilde{f_O}(\vec{k})\right\vert^2\\ \nonumber
&&D_n=\int d^2\vec{k}\left\vert\vec{k}\right\vert^n\left\vert\widetilde{f_O}(\vec{k})\right\vert^2
\end{eqnarray}
Note that $D_4=P_{40}+2P_{22}+P_{04}$. It is now clear that the method of Z08 basically utilizes the quadrupole moments in the Fourier space to measure the cosmic shear. 

Using eq.(\ref{f_O2}), we can find out how $P_{ij}$ transform under lensing:
\begin{eqnarray}
\label{pldl}
P_{ij}&=&\int d^2\vec{k}k_1^ik_2^j\left\vert\widetilde{W}_{\beta}(\vec{k})\vert \mathrm{det} (\mathbf{A})\vert\widetilde{f_S}(\mathbf{A}\vec{k})\right\vert^2\\ \nonumber
&=&\vert\mathrm{det} (\mathbf{A})\vert\int d^2\vec{k}(\mathbf{A}^{-1}\vec{k})_1^i(\mathbf{A}^{-1}\vec{k})_2^j\\ \nonumber
&\times&\left\vert\widetilde{W}_{\beta}(\mathbf{A}^{-1}\vec{k})\widetilde{f_S}(\vec{k})\right\vert^2
\end{eqnarray}
The last step in the above equation is achieved by redefining $\mathbf{A}\vec{k}$ as $\vec{k}$. For the isotropic Gaussian PSF, $\widetilde{W}_{\beta}(\vec{k})=\exp(-\beta^2\left\vert\vec{k}\right\vert^2/2)$. By keeping the terms in matrix $\mathbf{A}$ up to the first order in shear, one can straightforwardly show the following:
\begin{eqnarray}
\label{Pij_trans}
&&P_{20}-P_{02}\\ \nonumber
&=&P_{20}^S-P_{02}^S+2\kappa\beta^2\left(P_{40}^S-P_{04}^S\right)+4\gamma_2\beta^2\left(P_{31}^S-P_{13}^S\right)\\ \nonumber
&-&2\gamma_1\left[P_{20}^S+P_{02}^S-\beta^2\left(P_{40}^S+P_{04}^S-2P_{22}^S\right)\right]\\ \nonumber
\\ \nonumber
&&P_{11}\\ \nonumber
&=&P_{11}^S-\gamma_2\left(P_{20}^S+P_{02}^S-4\beta^2 P_{22}^S\right)+2\kappa\beta^2\left(P_{31}^S+P_{13}^S\right)\\ \nonumber
&+&2\gamma_1\beta^2\left(P_{31}^S-P_{13}^S\right)\\ \nonumber
\\ \nonumber
&&P_{20}+P_{02}\\ \nonumber
&=&P_{20}^S+P_{02}^S-2\gamma_1\left[P_{20}^S-P_{02}^S-\beta^2\left(P_{40}^S-P_{04}^S\right)\right]\\ \nonumber
&-&4\gamma_2\left[P_{11}^S-\beta^2\left(P_{31}^S+P_{13}^S\right)\right]+2\kappa\beta^2D_4^S\\ \nonumber
\\ \nonumber
&&D_4\\ \nonumber
&=&D_4^S-2\kappa\left(D_4^S-\beta^2D_6^S\right)\\ \nonumber
&-&2\gamma_1\left[2\left(P_{40}^S-P_{04}^S\right)-\beta^2\left(P_{60}^S+P_{42}^S-P_{24}^S-P_{06}^S\right)\right]\\ \nonumber
&-&4\gamma_2\left[2\left(P_{31}^S+P_{13}^S\right)-\beta^2\left(P_{51}^S+2P_{33}^S+P_{15}^S\right)\right]\\ \nonumber
\end{eqnarray}
where 
\begin{eqnarray}
\label{PS}
&&P_{ij}^S=\int d^2\vec{k}k_1^ik_2^j\left\vert\widetilde{W}_{\beta}(\vec{k})\widetilde{f_S}(\vec{k})\right\vert^2\\ \nonumber
&&D_n^S=\int d^2\vec{k}\left\vert\vec{k}\right\vert^n\left\vert\widetilde{W}_{\beta}(\vec{k})\widetilde{f_S}(\vec{k})\right\vert^2
\end{eqnarray}

By keeping the terms up to the first order in shear, and using the fact that the ensemble averages are taken over statistically isotropic galaxy samples, one can now directly find expressions for $\delta_1$ and $\delta_2$ of eq.(\ref{shear12PSF_wrong}):
\begin{eqnarray}
\label{deltas}
\delta_1&=&\left\langle\frac{\left( P_{20}^S-P_{02}^S\right)^2}{\left( P_{20}^S+P_{02}^S-\beta^2D_4^S/2\right)^2}\right\rangle_{en}\\ \nonumber
&-&2\beta^2\left\langle\frac{\left( P_{20}^S-P_{02}^S\right)\left( P_{40}^S-P_{04}^S\right)}{\left( P_{20}^S+P_{02}^S-\beta^2D_4^S/2\right)^2}\right\rangle_{en}\\ \nonumber
&+&\frac{\beta^4}{2}\left\langle\frac{\left( P_{20}^S-P_{02}^S\right)\left( P_{60}^S+P_{42}^S-P_{24}^S-P_{06}^S\right)}{\left( P_{20}^S+P_{02}^S-\beta^2D_4^S/2\right)^2}\right\rangle_{en}\\ \nonumber
\\ \nonumber
\delta_2&=&\left\langle\frac{4\left( P_{11}^S\right)^2}{\left( P_{20}^S+P_{02}^S-\beta^2D_4^S/2\right)^2}\right\rangle_{en}\\ \nonumber
&-&8\beta^2\left\langle\frac{P_{11}^S\left( P_{31}^S+P_{13}^S\right)}{\left( P_{20}^S+P_{02}^S-\beta^2D_4^S/2\right)^2}\right\rangle_{en}\\ \nonumber
&+&2\beta^4\left\langle\frac{P_{11}^S\left( P_{51}^S+2P_{33}^S+P_{15}^S\right)}{\left( P_{20}^S+P_{02}^S-\beta^2D_4^S/2\right)^2}\right\rangle_{en}
\end{eqnarray}

\label{lastpage}

\end{document}